\documentclass{article}

\usepackage{spconf,amsmath,graphicx}
\usepackage{graphicx}
\usepackage{amsmath,amssymb, amsthm, bbm}
\usepackage{epsfig}
\usepackage{varioref}
\usepackage{subfigure}
\usepackage{arydshln}
\usepackage{paralist}
\usepackage[colorlinks]{hyperref}
\usepackage{algpseudocode}
\usepackage{algorithm}
\usepackage[toc,page]{appendix}

\def\Csp{{\mathbb{C}}}

\def\Rsp{{\mathbb{R}}}

\def\Thetasp{{\mbox{\boldmath $\Theta$}}}

\def\Thetaspsc{{\mbox{\boldmath $\Thetasc$}}}

\def\bOmega{{\mbox{\boldmath $\Omega$}}}
\def\bOmegasc{{\mbox{\boldmath $\bOmega$}}}

\def\pxtheta{P_{\Xmatsc;\thetavecsc}}

\def\pxthetazero{P_{\Xmatsc;\thetavecsc_{0}}}

\def\qxtheta{Q^{\left(u\right)}_{\Xmatsc;\thetavecsc}}

\def\thetaveczero{\thetavec_{0}}
\def\thetaveczerosc{\thetavecsc_{0}}

\def\thetavec{{\bf{\theta}}}

\newcommand{\avec}{{\bf{a}}}

\newcommand{\yvec}{{\bf{y}}}

\newcommand{\xvec}{{\bf{x}}}

\newcommand{\hvec}{{\bf{h}}}

\newcommand{\etavec}{{\bf{\eta}}}

\newcommand{\zerovec}{{\bf{0}}}

\newcommand{\Gammamat}{{\bf{\Gamma}}}

\newcommand{\Amat}{{\bf{A}}}

\newcommand{\Bmat}{{\bf{B}}}

\newcommand{\Cmat}{{\bf{C}}}

\newcommand{\Fmat}{{\bf{F}}}

\newcommand{\Gmat}{{\bf{G}}}

\newcommand{\Imat}{{\bf{I}}}

\newcommand{\Rmat}{{\bf{R}}}

\newcommand{\Wmat}{{\bf{W}}}

\newcommand{\Xmat}{{\bf{X}}}

\newcommand{\Xmatsc}{{\mbox{\boldmath \tiny $\Xmat$}}}

\newcommand{\Zmatsc}{{\mbox{\boldmath \tiny $\Zmat$}}}

\newcommand{\Zmat}{{\bf{Z}}}

\def\bomega{{\mbox{\boldmath $\omega$}}}

\def\bOmega{{\mbox{\boldmath $\Omega$}}}

\def\bOmegasc{{\mbox{\boldmath \tiny $\bOmega$}}}

\def\bvartheta{{\mbox{\boldmath $\vartheta$}}}
\def\bvarthetasc{{\mbox{\boldmath \tiny $\vartheta$}}}

\def\bSigma{{\mbox{\boldmath $\Sigma$}}}

\def\bomega{{\mbox{\boldmath $\omega$}}}

\def\bGamma{{\mbox{\boldmath $\Gamma$}}}

\def\psivec{{\mbox{\boldmath $\psi$}}}

\def\tvec{{\mbox{\boldmath $t$}}}

\def\etavec{{\mbox{\boldmath $\eta$}}}

\def\thetavec{{\mbox{\boldmath $\theta$}}}

\def\thetavecsc{{\mbox{\boldmath \tiny {$\theta$}}}}

\def\etavecsc{{\mbox{\boldmath \tiny {$\eta$}}}}

\def\Thetasc{{\mbox{\tiny {$\Theta$}}}}

\def\XCalsc{{\mbox{\tiny $\mathcal{X}$}}}

\def\XCal{{\mbox{$\mathcal{X}$}}}

\def\muvec{{\mbox{\boldmath $\mu$}}}

\newcommand{\be}{\begin{equation}}

\newcommand{\ee}{\end{equation}}

\newcommand{\beqna}{\begin{eqnarray}}

\newcommand{\eeqna}{\end{eqnarray}}

\newtheorem{Theorem}{Theorem}
\newtheorem{Lemma}{Lemma}
\newtheorem{Corollary}{Corollary}

\newtheorem{Proposition}{Proposition}
\newtheorem{Definition}{Definition}

\newtheorem{Identity}{Identity}
\newtheorem{Relation}{Relation}

\title{Measure Transformed Quasi Score Test with Application to Location Mismatch Detection}
\name{Koby Todros}
\address{Ben-Gurion University of the Negev}
 
\begin{document}
\ninept
\maketitle
%%%%%%%%%%%%%%%%%%%%%%%%%%%%%%%%%%%%%%%%%%%%%%%%%%%%%%%%%%%%%%%%%%%%%%%%%%%%%%%%%%%%%%%%%%%%%%%%%%%
%%%%%%%%%%%%%%%%%%%%%%%%%%%%%%%%%%%%%%%%%%%%%%%%%%%%%%%%%%%%%%%%%%%%%%%%%%%%%%%%%%%%%%%%%%%%%%%%%%%
\begin{abstract}
In this paper, we develop a generalization of the Gaussian quasi score test (GQST) for composite binary hypothesis testing. The proposed test, called measure transformed GQST (MT-GQST), is based on the score-function of the measure transformed Gaussian quasi maximum likelihood estimator (MT-GQMLE) that operates by empirically fitting a Gaussian model to a transformed probability measure of the data. By judicious choice of the transform we show that, unlike the GQST, the proposed MT-GQST involves higher-order statistical moments and can gain resilience to outliers, leading to significant mitigation of the model mismatch effect on the decision performance. A data-driven procedure for optimal selection of the measure transformation parameters is developed that 
minimizes the spectral norm of the empirical asymptotic error-covariance of the MT-GQMLE. This amounts to maximization of an empirical worst-case asymptotic local power at a fixed asymptotic size. The MT-GQST is applied to location mismatch detection of a near-field point source in a simulation example that illustrates its robustness to outliers.
\end{abstract}
%%%
%%%%%%%%%%%%%%%%%%%%%%%%%%%%%%%%%%%%%%%%%%%%%%%%%%%%%%%%%%%%%%%%%%%%%%%%%%%%%%%%%%%%%%%%%%%%%%%%%%%
%%%%%%%%%%%%%%%%%%%%%%%%%%%%%%%%%%%%%%%%%%%%%%%%%%%%%%%%%%%%%%%%%%%%%%%%%%%%%%%%%%%%%%%%%%%%%%%%%%%
\section{Introduction}  
\label{sec:intro}
The score test \cite{Lehmann}-\cite{LM}, also known as Rao's score test, or the Lagrangian multiplier test, is a well established technique for composite binary hypothesis testing \cite{Lehmann}, \cite{Scharf}, \cite{Kay} that is based on the score-function\footnote{Score-function of an estimator is referred here to as the gradient of the estimator's objective function w.r.t. the vector parameter.}of the maximum likelihood estimator (MLE). Unlike the generalized likelihood ratio test (GLRT) \cite{Lehmann}, \cite{Kay} and Wald's test \cite{Lehmann}, \cite{Kay}, \cite{Wald}, it does not necessitate the maximum likelihood estimate under the alternative hypothesis, and therefore, may be significantly easier to compute. However, similarly to Wald's test and the GLRT, it assumes knowledge of the likelihood function. In many practical scenarios the likelihood function is unknown, and therefore, alternatives to the score test become attractive. 

A popular alternative of this kind is the Gaussian quasi score test (GQST) \cite{White}-\cite{bollerslev1992quasi} that is based on the score-function of the Gaussian quasi MLE (GQMLE) \cite{White}, \cite{bollerslev1992quasi}-\cite{Viberg}, which assumes that the samples obey Gaussian distribution. The GQST has gained popularity due to its implementation simplicity and ease of performance analysis that arise from the convenient Gaussian distribution. Despite the model mismatch, introduced by the normality assumption, the GQST has the appealing property of consistency under some mild regularity conditions \cite{WhiteBook}. However, in some circumstances, such as for certain types of non-Gaussian data, deviation from normality can inflict poor decision performance. This can occur when the first and second-order statistical moments are weakly identifiable over the parameter space, or in the case of heavy-tailed data when the non-robust sample mean and covariance provide poor estimates in the presence of outliers.

In this paper, a generalization of the GQST is developed. The proposed generalization, called measure transformed GQST (MT-GQST), is based on the score-function of the measure transformed GQMLE (MT-GQMLE) \cite{MTQML}, \cite{MTQMLFull} that operates by empirically fitting a Gaussian model to a transformed probability measure of the data. The considered measure-transformation, also applied in  \cite{MTCCA}-\cite{MTRT_B}, is structured by a non-negative function, called the MT-function, that weights the probability distribution of the data. By judicious choice of the MT-function we show that, unlike the GQST, the proposed MT-GQST involves higher-order statistical moments, can gain resilience to outliers, and yet have the computational and implementation advantages of the GQST.

Under some mild regularity conditions, we show that the MT-GQST is consistent. We also show that the asymptotic distribution of the test-statistic is central chi-squared under the null hypothesis, and non-central chi-squared under a sequence of local alternatives, with non-centrality parameter that is increasing with the inverse asymptotic error-covariance of the MT-GQMLE. 
A data driven procedure for optimal selection of the MT-function within some parametric class is developed that minimizes the spectral norm of the empirical asymptotic error-covariance of the MT-GQMLE. We show that this minimization amounts to maximization of an empirical worst-case asymptotic local power at a fixed asymptotic size. 

The  MT-GQST is illustrated for detecting a mismatch in the location of a near-field point source in the presence of spherically contoured noise. By specifying the MT-function within the family of zero-centred Gaussian functions parameterized by a scale parameter, we show that the MT-GQST outperforms the non-robust GQST and another robust alternative, and attains detection performance that are significantly closer to those of the omniscient score test that, unlike the MT-GQST, requires the knowledge of the likelihood function.

The paper is organized as follows. In Sections \ref{PMT}, the MT-GQMLE is reviewed. In Section \ref{MT_GQST}, the score-function of the MT-GQMLE is used to construct the proposed MT-GQST. The proposed test is applied to location mismatch detection in Section \ref{Examp}. In Section \ref{Conclusion}, the main points of this contribution are summarized. 
%%%%%%%%%%%%%%%%%%%%%%%%%%%%%%%%%%%%%%%%%%%%%%%%%%%%%%%%%%%%%%%%%%%%%%%%%%%%%%%%%%%%%%%%%%%%%%%%%%%%%%%%%%%%%%%%%%%%
%%%%%%%%%%%%%%%%%%%%%%%%%%%%%%%%%%%%%%%%%%%%%%%%%%%%%%%%%%%%%%%%%%%%%%%%%%%%%%%%%%%%%%%%%%%%%%%%%%%%%%%%%%%%%%%%%%%%%
\section{Measure transformed GQMLE: Review}
\label{PMT}
We begin by reviewing the principles of the parametric probability measure transform \cite{MTQML}, \cite{MTQMLFull}. We then define a parametric measure-transformed mean vector and covariance matrix and show their relation to higher-order statistical moments. Furthermore, we formulate their strongly consistent estimates and state conditions for outlier resilience. Finally, these quantities are used to construct the MT-GQMLE \cite{MTQML}, \cite{MTQMLFull}, whose objective function will be used in the following section to obtain the proposed MT-GQST. 
%%%%%%%%%%%%%%%%%%%%%%%%%%%%%%%%%%%%%%%%%%%%%%%%%%%%%%%%%%%%%%%%%%%%%%%%%%%%%%%%%%%%%%%%%%%%%%%%%%%%%%%%%%%%%%%%%%%%%
\subsection{Probability measure transform}
We define the measure space $\left(\XCal,\mathcal{S}_{\XCalsc},\pxtheta\right)$, where $\XCal\subseteq\Csp^{p}$ is the observation space of a random vector $\Xmat$, $\mathcal{S}_{\XCalsc}$ is a $\sigma$-algebra over $\XCal$ and $\pxtheta$ is an unknown probability measure on $\mathcal{S}_{\XCalsc}$ parameterized by a vector parameter $\thetavec$ that belongs to a parameter space $\Thetasp\subseteq\Rsp^{m}$.
%%%
\begin{Definition}  
\label{Def1}
Given a non-negative function $u:\Csp^{p}\rightarrow\Rsp_{+}$ satisfying 
\begin{equation} 
\label{Cond}
0<{{{\rm{E}}}\left[u\left(\Xmat\right);P_{\Xmatsc;\thetavecsc}\right]}<\infty,
\end{equation}
where
${\rm{E}}\left[u\left(\Xmat\right);\pxtheta\right]\triangleq\int_{\XCalsc}u\left(\xvec\right)d\pxtheta\left(\xvec\right)$
and $\xvec\in\XCal$, a transform on $P_{\Xmatsc;\thetavecsc}$ is defined via the relation:
\begin{equation}
\label{MeasureTransform} 
\qxtheta\left(A\right)\triangleq{\rm{T}}_{u}\left[\pxtheta\right]\left(A\right)=\int_{A}\varphi_{u}\left(\xvec;\thetavec\right)d\pxtheta\left(\xvec\right),
\end{equation}
where $A\in\mathcal{S}_{\XCalsc}$ and
$\varphi_{u}\left(\xvec;\thetavec\right)\triangleq{u\left(\xvec\right)}/{{{\rm{E}}}\left[u\left(\Xmat\right);\pxtheta\right]}$.
The function $u\left(\cdot\right)$ is called the MT-function. [Proof: see Appendix A in \cite{Todros2}]
\end{Definition}  
%%%
\begin{Proposition}[Properties of the transform]
\label{Prop1}
Let $\qxtheta$ be defined by relation (\ref{MeasureTransform}). 
Then
\begin{inparaenum}[1)]
\item
\label{P1}
$\qxtheta$ is a probability measure on $\mathcal{S}_{\XCalsc}$.
\item
\label{P2}
$\qxtheta$ is absolutely continuous w.r.t. $\pxtheta$, with Radon-Nikodym derivative \cite{Folland}:
\begin{equation}
\label{MeasureTransformRadNik}     
{d\qxtheta\left(\xvec\right)}/{d\pxtheta\left(\xvec\right)}=\varphi_{u}\left(\xvec;\thetavec\right).
\end{equation}
%%%
\end{inparaenum} 
\end{Proposition}
The probability measure $\qxtheta$ is said to be generated by the MT-function $u\left(\cdot\right)$. 
%By modifying $u\left(\cdot\right)$, such that the condition (\ref{Cond}) is satisfied, virtually any probability measure on $\mathcal{S}_{\XCalsc}$ can be obtained. 
%%%%%%%%%%%%%%%%%%%%%%%%%%%%%%%%%%%%%%%%%%%%%%%%%%%%%%%%%%%%%%%%%%%%%%%%%%%%%%%%%%%%%%%%%%%%%%%%%%%%%%%%%%%%%%%%%%%%%
\subsection{The MT-mean and MT-covariance}
According to (\ref{MeasureTransformRadNik}) the mean vector and covariance matrix of $\Xmat$ under $\qxtheta$, that are assumed to be known parameterized functions of $\thetavec$, are given by:
\begin{equation}
\label{MTMean} 
\muvec^{\left(u\right)}_{\Xmatsc;\thetavecsc}\triangleq{\rm{E}}[\Xmat\varphi_{u}\left(\Xmat;\thetavec\right);\pxtheta]
\end{equation}
and
\begin{equation}  
\label{MTCovZ}    
\bSigma^{\left(u\right)}_{\Xmatsc;\thetavecsc}\triangleq{\rm{E}}[\Xmat\Xmat^{H}\varphi_{u}\left(\Xmat;\thetavec\right);\pxtheta]-\muvec^{\left(u\right)}_{\Xmatsc;\thetavecsc}\muvec^{\left(u\right)H}_{\Xmatsc;\thetavecsc},
\end{equation}
respectively. Equations  (\ref{MTMean})  and (\ref{MTCovZ}) imply that $\muvec^{\left(u\right)}_{\Xmatsc;\thetavecsc}$ and $\bSigma^{\left(u\right)}_{\Xmatsc;\thetavecsc}$ are weighted mean and covariance of $\Xmat$ under $\pxtheta$, with the weighting function $\varphi_{u}\left(\cdot;\cdot\right)$ defined below (\ref{MeasureTransform}). By modifying the MT-function $u\left(\cdot\right)$, such that the condition (\ref{Cond}) is satisfied,  the MT-mean and MT-covariance under $\qxtheta$ are modified. In particular, by choosing $u\left(\cdot\right)$ to be any non-zero constant valued function we have $\qxtheta=\pxtheta$, for which the standard mean vector $\muvec_{\Xmatsc;\thetavecsc}$ and covariance matrix $\bSigma_{\Xmatsc;\thetavecsc}$ are obtained. Alternatively, when $u\left(\cdot\right)$ is non-constant analytic function, which has a convergent Taylor series expansion, the resulting MT-mean and MT-covariance involve higher-order statistical moments of $\pxtheta$.  
%%%%%%%%%%%%%%%%%%%%%%%%%%%%%%%%%%%%%%%%%%%%%%%%%%%%%%%%%%%%%%%%%%%%%%%%%%%%%%%%%%%%%%%%%%%%%%%%%%%%%%%%%%%%%%%%%%%%%
\subsection{The empirical MT-mean and MT-covariance}
\label{EmpMT}
Given a sequence of $N$ i.i.d. samples from $\pxtheta$ the empirical estimators of $\muvec^{\left(u\right)}_{\Xmatsc;\thetavecsc}$ and $\bSigma^{\left(u\right)}_{\Xmatsc;\thetavecsc}$ are defined as:
\begin{equation}   
\label{Mu_u_Est}  
\hat{\muvec}^{\left(u\right)}_{\Xmatsc}\triangleq\sum\nolimits_{n=1}^{N}\Xmat_{n}\hat{\varphi}_{u}\left(\Xmat_{n}\right)
\end{equation}
and
\begin{equation}    
\label{Rx_u_Est}      
\hat{\bSigma}^{\left(u\right)}_{\Xmatsc}\triangleq\sum\nolimits_{n=1}^{N}\Xmat_{n}\Xmat^{H}_{n}\hat{\varphi}_{u}\left(\Xmat_{n}\right)
-\hat{\muvec}^{\left(u\right)}_{\xvec}\hat{\muvec}^{\left(u\right)H}_{\xvec},
\end{equation}
respectively, where    
$\hat{\varphi}_{u}\left(\Xmat_{n}\right)\triangleq{u\left(\Xmat_{n}\right)}/{\sum_{n=1}^{N}u\left(\Xmat_{n}\right)}$.
According to Proposition 2 in \cite{Todros2}, if  ${\rm{E}}\left[\left\|\Xmat\right\|^{2}u\left(\Xmat\right);\pxtheta\right]<\infty$ then
$\hat{\muvec}^{\left(u\right)}_{\Xmatsc}\xrightarrow[N\rightarrow\infty]{\textrm{w.p. 1}}{\muvec}^{\left(u\right)}_{\Xmatsc;\thetavecsc}$ and $\hat{\bSigma}^{\left(u\right)}_{\Xmatsc}\xrightarrow[N\rightarrow\infty]{\textrm{w.p. 1}}{\bSigma}^{\left(u\right)}_{\Xmatsc;\thetavecsc}$, where ``$\xrightarrow{\textrm{w.p. 1}}$'' denotes convergence with probability (w.p.) 1 \cite{MeasureTheory}.

Robustness of the empirical MT-covariance (\ref{Rx_u_Est}) to outliers was studied in \cite{Todros2} using its influence function \cite{Hampel} which describes the effect on the estimator of an infinitesimal contamination at some point $\yvec\in\Csp^{p}$. An estimator is said to be B-robust if its influence function is bounded  \cite{Hampel}. Similarly to the proof of Proposition 3 in \cite{Todros2} it can be shown that if the MT-function $u(\yvec)$ and the product $u(\yvec)\|\yvec\|^{2}$ are bounded over $\Csp^{p}$ then the influence functions of both (\ref{Mu_u_Est}) and (\ref{Rx_u_Est}) are bounded.
%%%%%%%%%%%%%%%%%%%%%%%%%%%%%%%%%%%%%%%%%%%%%%%%%%%%%%%%%%%%%%%%%%%%%%%%%%%%%%%%%%%%%%%%%%%%%%%%%%%%%%%%%%%%%%%%%%%%%
%%%%%%%%%%%%%%%%%%%%%%%%%%%%%%%%%%%%%%%%%%%%%%%%%%%%%%%%%%%%%%%%%%%%%%%%%%%%%%%%%%%%%%%%%%%%%%%%%%%%%%%%%%%%%%%%%%%%%
\subsection{The MT-GQMLE}
\label{QMLEst}
Given a sequence of samples from $\pxtheta$, the MT-GQMLE \cite{MTQMLFull} of $\thetavec$ minimizes the empirical Kulback-Leibler divergence \cite{Kullback} between the transformed probability measure $Q^{\left(u\right)}_{\Xmatsc;\thetavecsc}$ and a complex circular Gaussian probability distribution $\Phi^{\left(u\right)}_{\Xmatsc;\bvarthetasc}$  \cite{Schreier}, characterized by the MT-mean $\muvec^{(u)}_{\Xmatsc;\bvarthetasc}$ and MT-covariance $\bSigma^{(u)}_{\Xmatsc;\bvarthetasc}$. We have shown that this minimization amounts to maximization of the objective function:
\begin{equation}
\label{ObjFun}
J_{u}\left(\bvartheta\right)\triangleq
-D[\hat{\bSigma}^{\left(u\right)}_{\Xmatsc}||\bSigma^{\left(u\right)}_{\Xmatsc;\bvarthetasc}] -
\|\hat{\muvec}^{\left(u\right)}_{\Xmatsc}-{\muvec}^{\left(u\right)}_{\Xmatsc;\bvarthetasc}\|^{2}_{{\bOmegasc^{(u)}_{\xvec;\bvarthetasc}}},
\end{equation}
where 
$D\left[\Amat||\Bmat\right]\triangleq{\rm{tr}}\left[\Amat\Bmat^{-1}\right]-\log\det\left[\Amat\Bmat^{-1}\right]-p$
is the log-determinant divergence \cite{LogDetDiv} between positive definite matrices $\Amat,\Bmat$, $\left\|\avec\right\|_{\Cmat}\triangleq\sqrt{\avec^{H}\Cmat\avec}$ denotes the weighted Euclidian norm of a vector $\avec$ with positive-definite weighting matrix $\Cmat$ and 
$\bOmega^{\left(u\right)}_{\Xmatsc;\bvarthetasc}\triangleq(\bSigma^{\left(u\right)}_{\Xmatsc;\bvarthetasc})^{-1}$. The MT-GQMLE is given by:
\begin{equation}
\label{PropEst}
\hat{\thetavec}_{u}=\arg\max\limits_{\bvarthetasc\in\Thetaspsc}J_{u}\left(\bvartheta\right).
\end{equation}
Under some mild regularity conditions, we have shown that the MT-GQMLE is asymptotically normal and unbiased with convergence rate of ${1/\sqrt{N}}$, i.e., 
%\begin{equation}
%\label{Asymp}
%\nonumber
$\sqrt{N}(\hat{\thetavec}_{u}-\thetavec)\xrightarrow{D}\mathcal{N}\left(\zerovec,\Rmat_{u}\left(\thetavec\right)\right)
\hspace{0.2cm}{\textrm{as}}\hspace{0.2cm}N\rightarrow\infty,
$
%\end{equation}
where ``$\xrightarrow{D}$'' denotes convergence in distribution \cite{MeasureTheory}. The asymptotic error-covariance is given by
\begin{equation}
\label{AMSE}  
\Rmat_{u}\left(\thetavec\right)=\Fmat^{-1}_{u}\left(\thetavec\right)\Gmat_{u}\left(\thetavec\right)\Fmat^{-1}_{u}\left(\thetavec\right),
\end{equation}
where
\begin{equation}   
\label{GDef}
\Gmat_{u}\left(\thetavec\right)\triangleq{\rm{E}}[u^2\left(\Xmat\right)\psivec_{u}\left(\Xmat;\thetavec\right)\psivec^{T}_{u}\left(\Xmat;\thetavec\right);P_{\Xmatsc;\thetavecsc}],
\end{equation}
\begin{equation}
\label{PsiDef}
\psivec_{u}\left(\Xmat;\thetavec\right)\triangleq\nabla_{\thetavecsc}\log\phi^{(u)}\left(\Xmat;\thetavec\right),
\end{equation}
%%%
\begin{equation}
\label{FDef}   
\Fmat_{u}\left(\thetavec\right)\triangleq-{\rm{E}}\left[u\left(\Xmat\right)\Gammamat_{u}\left(\Xmat;\thetavec\right);P_{\Xmatsc;\thetavecsc}\right],
\end{equation}
%%%  
\begin{equation}
\label{GammaDef}  
\Gammamat_{u}\left(\Xmat;\thetavec\right)\triangleq\nabla^{2}_{\thetavecsc}\log\phi^{(u)}\left(\Xmat;\thetavec\right),
\end{equation}
$\phi^{(u)}\left(\Xmat;\thetavec\right)$ is the density  of the Gaussian measure $\Phi^{\left(u\right)}_{\Xmatsc;\thetavecsc}$ and it is assumed that $\Fmat_{u}\left(\thetavec\right)$ is non-singular.
%%%%%%%%%%%%%%%%%%%%%%%%%%%%%%%%%%%%%%%%%%%%%%%%%%%%%%%%%%%%%%%%%%%%%%%%%%%%%%%%%%%%%%%%%%%%%%%%%%%%%%%%%%%%%%%%%%%%%%%%%%%%%%%%%%%%%%%%%%%%%%%%%%%%%%%%%%%%%%%%%%%%%%%%%%%%%%%%%%%%%%%%%%%%%%%%%%%%%%
\section{The Measure transformed Gaussian quasi score test}
\label{MT_GQST}
Given a sequence of samples from $\pxtheta$, we use the score-function of  the MT-GQMLE (\ref{PropEst}) to construct the proposed MT-GQST for the composite hypothesis testing problem:
\begin{eqnarray}
\label{CompHyp}
H_{0}:&\thetavec=\thetavec_{0}
\\\nonumber
H_{1}:&\thetavec\neq\thetavec_{0}.
\end{eqnarray}
Under some mild regularity conditions we show that the MT-GQST is consistent. Furthermore, we derive the asymptotic distribution of its test-statistic under the null hypothesis and under a sequence of local alternatives. Finally, a data driven procedure for optimal selection of the MT-function is developed.  
%%%%%%%%%%%%%%%%%%%%%%%%%%%%%%%%%%%%%%%%%%%%%%%%%%%%%%%%%%%%%%%%%%%%%%%%%%%%%%%%%%%%%%%%%%%%%%%%%%%
\subsection{The MT-GQST}
Notice that by (\ref{Mu_u_Est}) and (\ref{Rx_u_Est}) the objective function (\ref{ObjFun}) is an empirical estimate of $\bar{J}_{u}\left(\thetavec,\bvartheta\right)\triangleq-D[{\bSigma}^{(u)}_{\Xmatsc;\thetavecsc}||\bSigma^{(u)}_{\Xmatsc;\bvarthetasc}] -
\|{\muvec}^{(u)}_{\Xmatsc;\thetavecsc}-{\muvec}^{(u)}_{\Xmatsc;\bvarthetasc}\|^{2}_{{\bOmegasc^{(u)}_{\xvec;\bvarthetasc}}}$. One can verify that  when $\bar{J}_{u}\left(\thetavec,\bvartheta\right)$ is $\bvartheta$-differentiable it has a stationary point at $\bvartheta=\thetavec$. Assuming that this stationary point is unique $\nabla_{\bvarthetasc}\bar{J}_{u}\left(\thetavec,\thetaveczero\right)=\zerovec$ when $\thetavec=\thetaveczero$ and $\nabla_{\bvarthetasc}\bar{J}_{u}\left(\thetavec,\thetaveczero\right)\neq\zerovec$ when $\thetavec\neq\thetaveczero$. This motivates the use of $\nabla{J}_{u}\left(\thetavec\right)$, i.e., the score-function of the MT-GQMLE, for testing between $H_{0}$ and $H_{1}$. Hence, we define the normalized score-function: ${\etavec}_{u}\left(\thetavec\right)\triangleq\nabla{J}_{u}\left(\thetavec\right)\times(1/\sqrt{N})\sum_{n=1}^{N}u\left(\Xmat_{n}\right)$. By (\ref{Mu_u_Est})-(\ref{ObjFun})  
\begin{equation}
\label{TildeJ}
{\etavec}_{u}\left(\thetavec\right)=(1/\sqrt{N})\sum\nolimits_{n=1}^{N}u\left(\Xmat_{n}\right)\psivec_{u}\left(\Xmat_{n};\thetavec\right),
\end{equation}
where $\psivec_{u}\left(\Xmat;\thetavec\right)$ is defined in (\ref{PsiDef}). Furthermore, we define the empirical estimate of (\ref{GDef}):
\begin{equation}
\label{hatG}  
\hat{\Gmat}_{u}(\thetavec)\triangleq{N}^{-1}\sum\nolimits_{n=1}^{N}u^{2}\left(\Xmat_{n}\right)\psivec_{u}\left(\Xmat_{n};\thetavec\right)\psivec^{T}_{u}\left(\Xmat_{n};\thetavec\right).
\end{equation}
The MT-GQST for the hypothesis testing problem (\ref{CompHyp}) is defined as:
\begin{equation}
\label{MTRAO}
T_{u}\triangleq{\etavec^{T}_{u}\left(\thetaveczero\right)\hat{\Gmat}^{-1}_{u}(\thetaveczero)\etavec_{u}\left(\thetaveczero\right)}\mathop{\gtreqless}_{H_{0}}^{H_{1}}t,
\end{equation}
where $t\in\Rsp_{+}$ denotes a threshold.  By modifying the MT-function $u\left(\cdot\right)$ such that condition (\ref{Cond}) is
satisfied the MT-GQST is modified, resulting in a family of tests. In particular, when $u\left(\cdot\right)$ is any non-zero constant function $\qxtheta=\pxtheta$
and the standard GQST is obtained that only involves first and second-order statistical moments. 
%%%%%%%%%%%%%%%%%%%%%%%%%%%%%%%%%%%%%%%%%%%%%%%%%%%%%%%%%%%%%%%%%%%%%%%%%%%%%%%%%%%%%%%%%%%%%%%%%%%
%%%%%%%%%%%%%%%%%%%%%%%%%%%%%%%%%%%%%%%%%%%%%%%%%%%%%%%%%%%%%%%%%%%%%%%%%%%%%%%%%%%%%%%%%%%%%%%%%%%
\subsection{Asymptotic performance analysis}
Here, we study the asymptotic performance of the proposed test (\ref{MTRAO}). We assume that a sequence of i.i.d. samples $\Xmat_{n}$,  $n=1,\ldots,N$ from $\pxtheta$ is available and that the parameter space $\Thetasp$ is compact. We begin by stating some regularity conditions that will be used in the sequel:
\begin{enumerate}[({A}-1)]
\item
\label{A1}
${\rm{E}}\left[u\left(\Xmat\right)\psivec_{u}\left(\Xmat;\thetaveczero\right);\pxtheta\right]\neq\zerovec$ for $\thetavec\neq\thetaveczero$.
\item
\label{A2} 
${\rm{E}}[u^2\left(\Xmat\right)\psivec_{u}\left(\Xmat;\thetaveczero\right)\psivec^{T}_{u}\left(\Xmat;\thetaveczero\right);P_{\Xmatsc;\thetavecsc}]$ is non-singular.
\item
\label{A3}  %B2, C2
$\muvec^{\left(u\right)}_{\Xmatsc;\thetavecsc}$ and $\bSigma^{\left(u\right)}_{\Xmatsc;\thetavecsc}$ are twice continuously differentiable in $\Thetasp$.
\item  
\label{A4} % A4 B3  
${\rm{E}}\left[u^{4}\left(\Xmat\right);\pxtheta\right]$ and ${\rm{E}}\left[\left\|\Xmat\right\|^{8}u^{4}\left(\Xmat\right);\pxtheta\right]$ are bounded.
\item
\label{A5} % B1
$\Gmat_{u}\left(\thetavec\right)$ is bounded and non-singular.
%%%
\item
\label{A6}
The density of $\pxtheta$ is continuous in $\Thetasp$ a.e. over $\XCal$.
\item
\label{A7}
The Fisher information matrix  $\Imat_{\rm{FIM}}\left(\thetavec\right)$ \cite{Schreier} is bounded.
\end{enumerate}
%%%

The following proposition states consistency conditions.
\begin{Proposition}[Consistency]
\label{ConsProp}
Assume that conditions A-\ref{A1}$-$A-\ref{A4} are satisfied.
%%%
Then, for any $t\in\Rsp$
\begin{equation}
\label{InfLim}
{\rm{Pr}}\left[T_{u}>t\right]\xrightarrow[N\rightarrow{\infty}]{}1\hspace{0.2cm}{\text{under $H_{1}$}}.
\end{equation}
[A proof is given in Appendix \ref{ConsPropProof}]
\end{Proposition}
%%%
Next, we derive the asymptotic distribution of the test-statistic under the null hypothesis and under a sequence of local alternatives.
\begin{Proposition}[Asymptotic distribution under the null hypothesis]
\label{AChiH0}  
Assume that conditions A-\ref{A3}$-$A-\ref{A5} are satisfied.
Then, 
\begin{equation}
\label{TAsDistH0}
T_{u}\xrightarrow[N\rightarrow{\infty}]{D}{\chi}^{2}_{m}\hspace{0.2cm}\text{under $H_{0}$},
\end{equation}
where ${\chi}^{2}_{m}$ denotes a central chi-squared distribution with $m$-degrees of freedom. 
[A proof appears in Appendix \ref{AChiH0Proof}]
\end{Proposition}
%%%
%%%
\begin{Theorem}[Asymptotic distribution under local alternatives]
\label{ALADTh}
Assume that conditions A-\ref{A3}$-$A-\ref{A7} are satisfied.
Furthermore, consider a sequence of local alternatives that converges to $\thetaveczero$ at a rate of $1/\sqrt{N}$. Specifically, consider
\begin{equation}
\label{LocAlt}
H_{1}: \thetavec=\thetaveczero+\hvec/\sqrt{N},
\end{equation}
where $\hvec\in\Rsp^{m}$ is a non-zero locality parameter. Then, 
\begin{equation}
\label{AsDistH1}
T_{u}\xrightarrow[N\rightarrow{\infty}]{D} {\chi}^{2}_{m}\left(\lambda_{u}(\hvec)\right)\hspace{0.2cm}\text{under $H_{1}$},
\end{equation}
where ${\chi}^{2}_{m}\left(\lambda_{u}(\hvec)\right)$ is a non-central chi-squared distribution with $m$-degrees of freedom and non-centrality parameter $\lambda_{u}(\hvec)\triangleq\hvec^{T}{\Rmat}^{-1}_{u}({\thetaveczero})\hvec$. The matrix $\Rmat_{u}(\cdot)$ is the asymptotic error-covariance (\ref{AMSE}) of the MT-GQMLE (\ref{PropEst}). [A proof appears in Appendix \ref{ALADThProof}]
\end{Theorem}
The following Corollary is a direct consequence of (\ref{AsDistH1}), the Rayleigh-Ritz Theorem \cite{Horn} and the property that the right-tail probability
of the non-central chi-squared distribution is monotonically increasing in the non-centrality parameter \cite{Aubel}.
\begin{Corollary}[Asymptotic local power]
\label{APow}
Assume that conditions A-\ref{A3}$-$A-\ref{A7} hold. Under the local alternatives (\ref{LocAlt}), the asymptotic power at a fixed asymptotic size $\alpha$ satisfies
\begin{equation}
\label{ALocPow}
\beta^{(\alpha)}_{u}\left(\hvec\right)={Q}_{\chi^{2}_{m}\left(\lambda_{u}(\hvec)\right)}\left({Q}^{-1}_{\chi^{2}_{m}}\left(\alpha\right)\right),
\end{equation}
where ${Q}_{\chi^{2}_{m}}\left(\cdot\right)$ and ${Q}_{\chi^{2}_{m}\left(\cdot\right)}\left(\cdot\right)$ denote the right-tail probabilities of the central and non-central chi-squared distributions, respectively.
Furthermore, for any $c>0$ the worst-case asymptotic power
\begin{equation}
\label{WCAP} 
{\bar{\beta}}^{(\alpha)}_{u}(c)\triangleq\min\limits_{\hvec:\|\hvec\|\geq{c}}\beta^{(\alpha)}_{u}(\hvec)={Q}_{\chi^{2}_{m}\left(\gamma_{u}(c)\right)}\left({Q}^{-1}_{\chi^{2}_{m}}\left(\alpha\right)\right),
\end{equation}
where $\gamma_{u}(c)\triangleq{c}^{2}\left\|{\Rmat}_{u}({\thetaveczero})\right\|^{-1}_{S}$ and $\left\|\cdot\right\|_{S}$ denotes the spectral norm.
\end{Corollary}
%%%%%%%%%%%%%%%%%%%%%%%%%%%%%%%%%%%%%%%%%%%%%%%%%%%%%%%%%%%%%%%%%%%%%%%%%%%%%%%%%%%%%%%%%%%%%%%%%%%%%%%%%%%%%%%%%%%%%%%%%%%%%%%%%%%%%%%%%%%%%%%%%%%%%%%%%%%%%%%%%%%%%%%%%%%%%%%%%%%%%%%%%%%%%%%%%%%%%%
\subsection{Selection of the MT-function}
\label{OptChoice}
While according to Propositions \ref{ConsProp} and \ref{AChiH0}, the asymptotic global power and size are invariant to the choice of the MT-function $u(\cdot)$, by Corollary \ref{APow} one sees that it controls the asymptotic local power through the error-covariance  ${\Rmat}_{u}(\thetaveczero)$ (\ref{AMSE}). Since the tail probability of the non-central chi-squared distribution is monotonically increasing in the non-centrality parameter \cite{Aubel}, minimization of the spectral norm ${\|{\Rmat}_{u}(\thetaveczero)\|}_{S}$ amounts to maximization of the worst case asymptotic local power ${\bar{\beta}}^{(\alpha)}_{u}(c)$ (\ref{WCAP}) for any fixed $c$ and asymptotic size $\alpha$. Hence, we propose to choose $u(\cdot)$ that minimizes ${\|\hat{\Rmat}_{u}(\thetaveczero)\|}_{S}$, where $\hat{\Rmat}_{u}(\thetavec)$ is an empirical estimate of error-covariance (\ref{AMSE}) defined as:
\begin{equation}
\label{EMSE}  
\hat{\Rmat}_{u}(\thetavec)\triangleq{{\hat{\Fmat}}^{-1}_{u}(\thetavec)\hat{\Gmat}_{u}(\thetavec){\hat{\Fmat}}^{-1}_{u}(\thetavec)},
\end{equation}
where
$
\hat{\Fmat}_{u}\left(\thetavec\right)\triangleq-{N}^{-1}\sum_{n=1}^{N}u\left(\Xmat_{n}\right)\Gammamat_{u}\left(\Xmat_{n};\thetavec\right)
$
is an estimate of (\ref{FDef}) and $\hat{\Gmat}_{u}\left(\thetavec\right)$ is defined in (\ref{hatG}). It can be shown that if conditions A-\ref{A3}, A-\ref{A4}, A-\ref{A5}, A-\ref{A7}, are satisfied then $\hat{\Rmat}_{u}(\thetaveczero)\xrightarrow[N\rightarrow{\infty}]{P}{\Rmat}_{u}(\thetaveczero)$ under the local alternatives (\ref{LocAlt}).

Here, we restrict the class of MT-functions to some parametric family $\left\{u\left(\Xmat;\bomega\right), \bomega\in\bOmega\subseteq\Csp^{r}\right\}$ that satisfies the conditions stated in Definition \ref{Def1} and Theorem \ref{ALADTh}. For example, the Gaussian family of functions that satisfy the robustness conditions stated at the ending paragraph of Subsection \ref{EmpMT} is a natural choice for inducing outlier resilience. Hence, an optimal choice of the MT-function parameter $\bomega$ would minimize ${\|\hat{\Rmat}_{u}(\thetaveczero)\|}_{S}$ that is constructed from (\ref{EMSE}) by the same sequence of samples used for obtaining the MT-GQST (\ref{MTRAO}).
%%%%%%%%%%%%%%%%%%%%%%%%%%%%%%%%%%%%%%%%%%%%%%%%%%%%%%%%%%%%%%%%%%%%%%%%%%%%%%%%%%%%%%%%%%%%%%%%%%%%%%%%%%%%%%%%%%%%%%%%%%%%%%%%%%%%%%%%%%%%%%%%%%%%%%%%%%%%%%%%%%%%%%%%%%%%%%%%%%%%%%%%%%%%%%%%%%%%%%%%%%%%%%%%%%%%%%%%%%%%%%%%%%%%%%%%%%
\section{Numerical example}
\label{Examp}
%%%%%%%%%%%%%%%%%%%%%%%%%%%%%%%%%%%%%%%%%%%%%%%%%%%%%%%%%%%%%%%%%%%%%%%%%%%%%%%%%%%%%%%%%%%%%%%%%%%%%%%%%%%%%%%%%%%%%
We consider the problem of detecting a mismatch in the location of an emitting narrowband near-field point source that is formulated as the following composite binary hypothesis testing problem:
\begin{eqnarray}
\label{CompHypMis}
H_{0}&:&
\Xmat_{n}=S_{n}\avec\left(\thetaveczero\right)+\Wmat_{n},\hspace{0.2cm}n=1,\ldots,N,
\\\nonumber
H_{1}&:&
\Xmat_{n}=S_{n}\avec\left(\thetavec\right)+\Wmat_{n},\hspace{0.2cm}\thetavec\neq\thetaveczero,\hspace{0.2cm}n=1,\ldots,N,
\end{eqnarray}
where $\{\Xmat_{n}\in\Csp^{p}\}$ is an observation process, $\{S_{n}\in\Csp\}$ is an i.i.d. symmetrically distributed random signal process, and $\{\Wmat_{n}\in\Csp^{p}\}$ is an i.i.d noise process that is statistically independent of $\{S_{n}\}$. We assume that the noise component is spherically contoured \cite{Visa} with stochastic representation 
%\begin{equation}
%\label{CompGauss}
$\Wmat_{n}=\nu_{n}\Zmat_{n},$
%\end{equation}
where $\{\nu_{n}\in\Rsp_{++}\}$ is an i.i.d. process and $\{\Zmat_{n}\in\Csp^{p}\}$ is a proper-complex wide-sense stationary Gaussian process with zero-mean and scaled unit covariance $\sigma^{2}_{\Zmatsc}\Imat$. The processes $\{\nu_{n}\}$ and $\{\Zmat_{n}\}$ are assumed to be statistically independent. The vector $\avec\left(\thetavec\right)$, $\thetavec\triangleq\left[r,\vartheta\right]^{T}$, is the steering vector of a uniform linear array of $p$ sensors with inter-element spacing $d$ that receive a signal with wavelength $\lambda$ generated by a narrowband near-field point source with range $r$ and bearing $\vartheta$. By Fresnel's approximation \cite{Fresnel}, \cite{Fresnel2} when $0.62(d^{3}{(p-1)^{3}}/{\lambda})^{1/2}<r<2d^{2}{{(p-1)^{2}}/{\lambda}}$ we have $\left[\avec\left(\thetavec\right)\right]_{k}=\exp\left(j(\omega_{\rm{e}}{k}+\phi_{\rm{e}}{k}^{2} + O(d^{2}/r^{2}))\right)$, $k=0,\ldots,p-1$, where $\omega_{\rm{e}}\triangleq-{2\pi{d}\sin\left(\vartheta\right)}/{\lambda}$ and $\phi_{\rm{e}}\triangleq{\pi{d}^{2}\cos^{2}(\vartheta)}/({\lambda{r}})$ are called ``electrical angles''. 

In order to gain robustness against outliers, as well as sensitivity to higher-order moments, we specify the MT-function in the zero-centred Gaussian family of functions parametrized by a width parameter $\omega$, i.e.,
\begin{equation}
\label{GaussMTFunc}
u\left(\xvec;\omega\right)=\exp\left(-{\left\|\xvec\right\|^{2}}/{\omega^{2}}\right),\hspace{0.2cm\omega\in\Rsp_{++}}.
\end{equation}
Notice that the MT-function (\ref{GaussMTFunc}) satisfies the robustness conditions stated at the ending paragraph of Subsection \ref{EmpMT}. To obtain the corresponding test-statistic (\ref{MTRAO}) and the empirical error-covariance (\ref{EMSE}), one has to compute the vector and matrix functions $\psivec_{u}\left(\cdot;\thetavec\right)$ and $\bGamma_{u}\left(\cdot;\thetavec\right)$, defined in (\ref{PsiDef})  and  (\ref{GammaDef}), respectively. By (\ref{MTMean}), (\ref{MTCovZ}), (\ref{PsiDef}), (\ref{GammaDef}), (\ref{CompHypMis}) and (\ref{GaussMTFunc}) these quantities take the following simple forms:
$$
\psivec_{u}\left(\Xmat;\thetavec\right)=\xi\left(\omega\right)\nabla_{\thetavecsc}|\Xmat^{H}\avec\left(\thetavec\right)|^{2},
$$
$$
\Gammamat_{u}\left(\Xmat;\thetavec\right)=\xi\left(\omega\right)\left[{\begin{array}{*{20}c} 
\nabla^{2}_{r}\left|\Xmat^{H}\avec\left(\thetavec\right)\right|^{2} & \nabla^{2}_{r\vartheta}\left|\Xmat^{H}\avec\left(\thetavec\right)\right|^{2}
\\
\nabla^{2}_{r\vartheta}\left|\Xmat^{H}\avec\left(\thetavec\right)\right|^{2} & \nabla^{2}_{\vartheta}\left|\Xmat^{H}\avec\left(\thetavec\right)\right|^{2} \end{array}}\right],
$$
where $\xi\left(\omega\right)$ is a strictly positive functions of $\omega$. It is important to note that the resulting test-statistic (\ref{MTRAO}) and the empirical error-covariance (\ref{EMSE}) are \textit{independent} of $\xi\left(\omega\right)$. 

In the following simulation we evaluate the detection performance of the proposed MT-GQST as compared to the score test, the GQST, and another robust GQST extension, called here ZMNL-GQST. The ZMNL-GQST operates by applying GQST after passing the data through a zero-memory non-linear (ZMNL) function that suppresses outliers by clipping the amplitude of the observations. We use the same ZMNL preprocessing approach that has been applied in \cite{swami1997}-\cite{swami2002} to robustify the MUSIC algorithm \cite{Viberg} . 

We considered a BPSK signal with variance $\sigma^{2}_{S}$ impinging on $p=8$ sensors with inter-element spacing $d=\lambda/4=0.25$ [m]. Two types of noise distributions were examined: 1) Gaussian and 2) heavy-tailed $K$-distributed noise \cite{Visa} with shape parameter $\kappa=0.75$. The sample size was set to $N=1000$. The signal-to-noise-ratio (SNR), used to index the detection performance, is defined as ${\rm{SNR}}\triangleq10\log_{10}{\sigma^{2}_{S}}/{\sigma^{2}_{\Zmatsc}}$. The location vector parameter at the null was set to $\thetavec_{0}=[r_{0},\vartheta_{0}]^{T}$, where $r_{0}=1.5$ [m] and $\vartheta_{0}=0^{\circ}$. We considered a specific local alternative $\thetavec_{1}=[r_{1},\vartheta_{1}]^{T}$, corresponding to $\hvec=\sqrt{N}\left(\thetavec_{1}-\thetavec_{0}\right)$ in (\ref{LocAlt}), where $r_{1}=r_0 + 0.01$ [m] and $\vartheta_{1}=\vartheta_{0} + 0.5^{\circ}$. The optimal width parameter $\omega_{\rm{opt}}$ of Gaussian MT-function (\ref{GaussMTFunc}) was obtained by minimizing the spectral norm ${\|\hat{\Rmat}_{u}(\thetaveczero)\|}_{S}$ of the empirical error-covariance (\ref{EMSE}) over $\Omega=[1, 30]$. All empirical power curves were obtained using $10^{4}$ Monte-Carlo simulations. 

Fig. \ref{Fig1} depicts the empirical and asymptotic (\ref{ALocPow}) power curves of the MT-GQST as compared to the empirical power curves of the GQST, ZMNL-GQST and the score test for a fixed test size equal to $10^{-2}$. Notice that when the noise is Gaussian, the MT-GQST, GQST and ZMNL-GQST attain similar performance. For the $K$-distributed noise, the MT-GQST outperforms the GQST and ZMNL-GQST, and significantly reduces the gap towards the score test, which unlike the MT-GQST, assumes knowledge of the likelihood function. 
 %%%
 \begin{figure}[htbp!]  
  \begin{center}
    {{\subfigure{\label{Fig1c}\includegraphics[width=8cm, height=3.821cm]{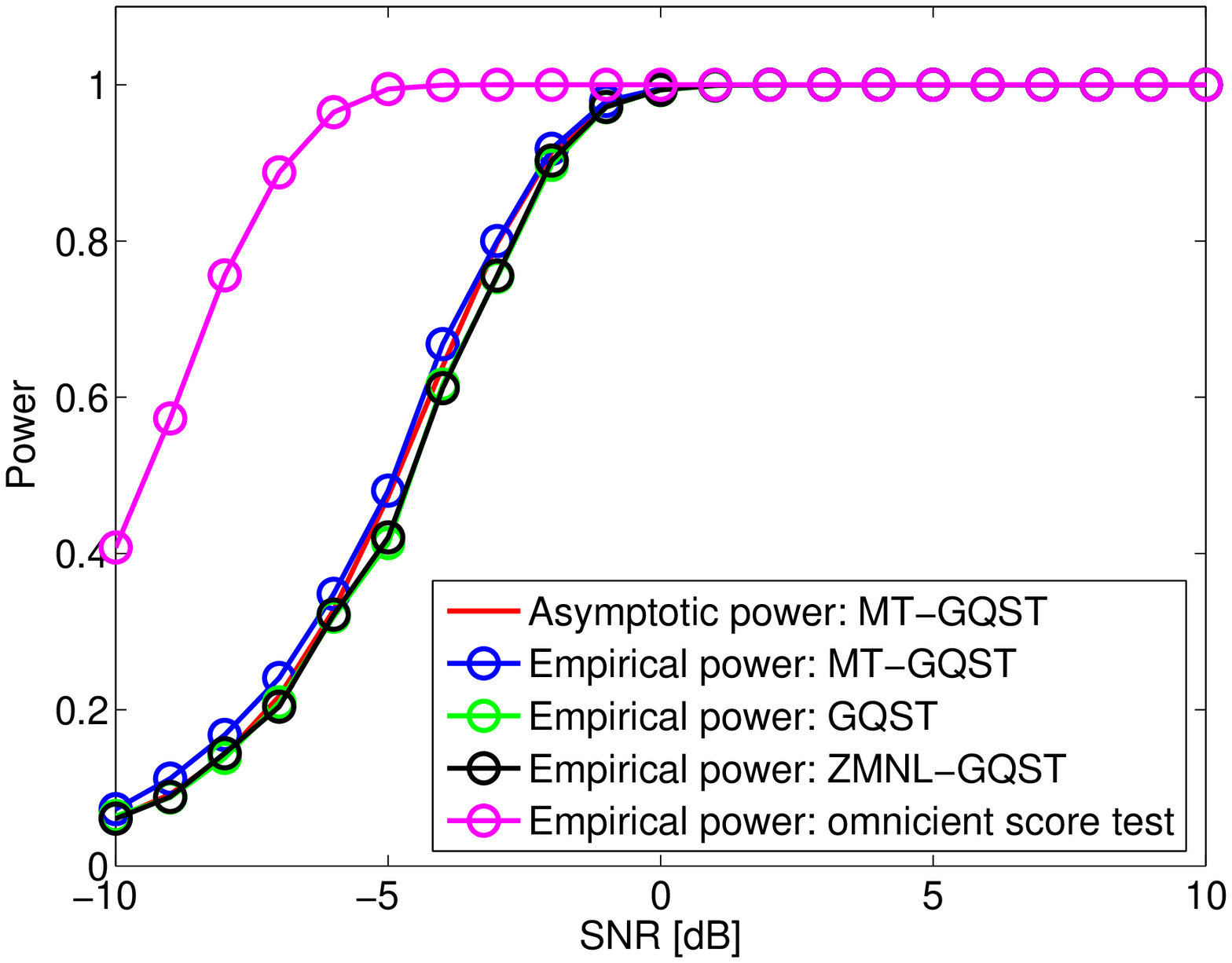}}}}
    {{\subfigure{\label{Fig1c}\includegraphics[width=8cm, height=3.821cm]{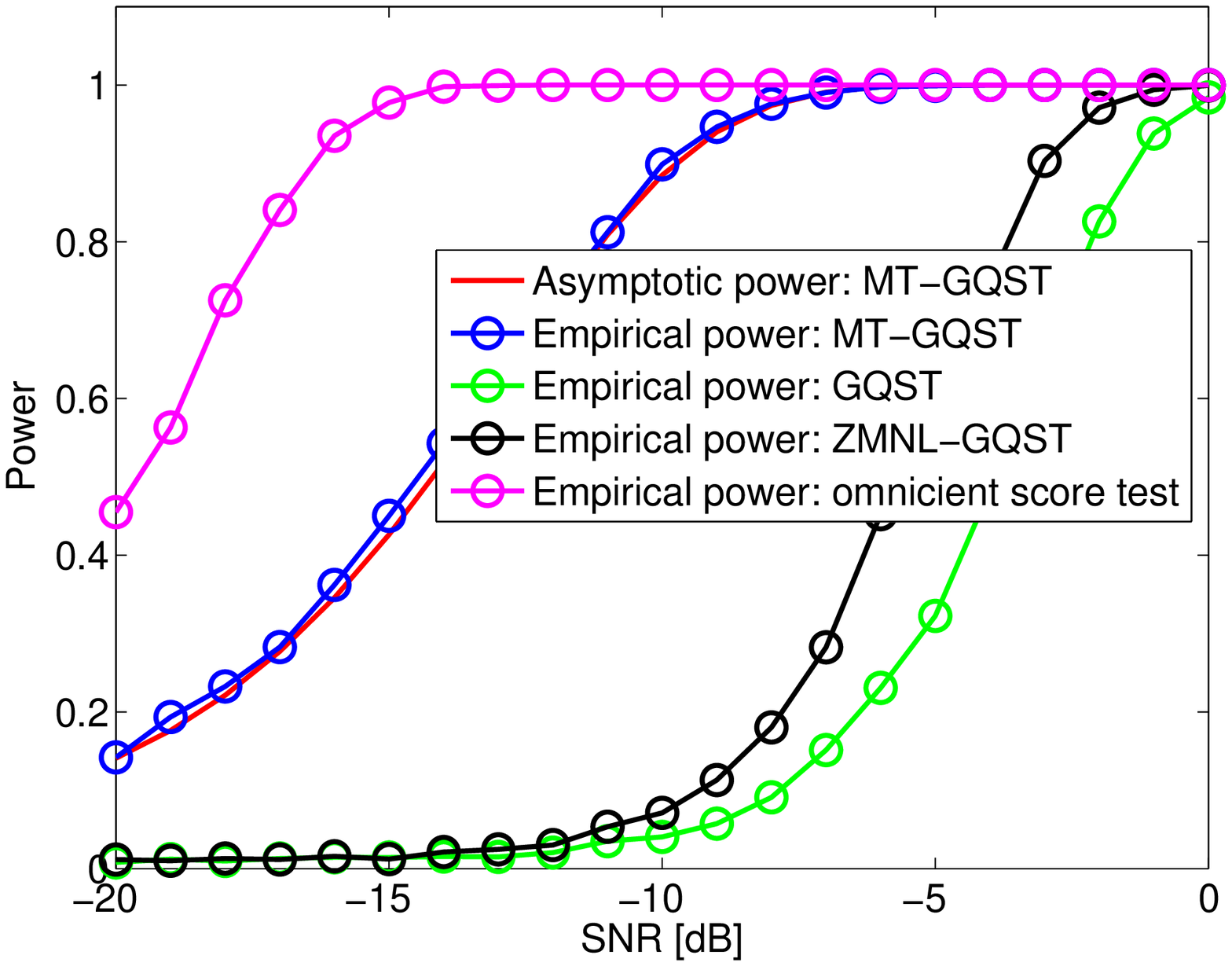}}}}
      \end{center}  
  \caption{Location mismatch detection in Gaussian noise \textbf{(top)}  and $K$-distributed noise \textbf{(bottom)}.}   
 \label{Fig1}
\end{figure}
%%%%%%%%%%%%%%%%%%%%%%%%%%%%%%%%%%%%%%%%%%%%%%%%%%%%%%%%%%%%%%%%%%%%%%%%%%%%%%%%%%%%%%%%%%%%%%%%%%%%%%%%%%%%%%%%%%%%
%%%%%%%%%%%%%%%%%%%%%%%%%%%%%%%%%%%%%%%%%%%%%%%%%%%%%%%%%%%%%%%%%%%%%%%%%%%%%%%%%%%%%%%%%%%%%%%%%%%%%%%%%%%%%%%%%%%%
\section{Conclusion}
\label{Conclusion}
In this paper a new score-type test, called MT-GQST, was derived based on the score-function of the measure transformed GQMLE. By specifying the MT-function in the Gaussian family, the proposed test was applied to location mismatch detection in non-Gaussian noise. Exploration of other MT-functions may result in additional tests in this class that have different useful properties. 
%%%%%%%%%%%%%%%%%%%%%%%%%%%%%%%%%%%%%%%%%%%%%%%%%%%%%%%%%%%%%%%%%%%%%%%%%%%%%%%%%%%%%%%%%%%%%%%%%%%%%%%%%%%%%%%%%%%%
%%%%%%%%%%%%%%%%%%%%%%%%%%%%%%%%%%%%%%%%%%%%%%%%%%%%%%%%%%%%%%%%%%%%%%%%%%%%%%%%%%%%%%%%%%%%%%%%%%%%%%%%%%%%%%%%%%%%

\appendix
\section*{Appendices}
\addcontentsline{toc}{section}{Appendices}
\renewcommand{\thesubsection}{\Alph{subsection}}

%%%%%%%%%%%%%%%%%%%%%%%%%%%%%%%%%%%%%%%%%%%%%%%%%%%%%%%%%%%%%%%%%%%%%%%%%%%%%%%%%%%%%%%%%%%%%%%%%%%%%%%%%%%%%%%%%%%%
%%%%%%%%%%%%%%%%%%%%%%%%%%%%%%%%%%%%%%%%%%%%%%%%%%%%%%%%%%%%%%%%%%%%%%%%%%%%%%%%%%%%%%%%%%%%%%%%%%%%%%%%%%%%%%%%%%%%
\subsection{Proof of Proposition \ref{ConsProp}:}
\label{ConsPropProof}
We first show that ${\etavec}_{u}\left(\thetaveczero\right)$ and $\hat{\Gmat}_{u}(\thetaveczero)$ satisfy
\begin{equation}
\label{ScoreLimit}
\frac{1}{\sqrt{N}}{\etavec}_{u}\left(\thetaveczero\right)\xrightarrow[N\rightarrow{\infty}]{\textrm{w.p. 1}}{\rm{E}}\left[u\left(\Xmat\right)\psivec_{u}\left(\Xmat;\thetaveczero\right);\pxtheta\right]
\end{equation}
and 
\begin{equation}
\label{CovLimit}
\hat{\Gmat}_{u}(\thetaveczero)\xrightarrow[N\rightarrow{\infty}]{\textrm{w.p. 1}}{\rm{E}}[u^2\left(\Xmat\right)\psivec_{u}\left(\Xmat;\thetaveczero\right)\psivec^{T}_{u}\left(\Xmat;\thetaveczero\right);P_{\Xmatsc;\thetavecsc}].
\end{equation}
Since $\left\{\Xmat_{n}\right\}^{N}_{n=1}$ are i.i.d. random variables and the functions $u\left(\cdot\right)$ and $\psivec_{u}\left(\cdot,\cdot\right)$ are real, the products $\left\{u\left(\Xmat_{n}\right)\psivec_{u}\left(\Xmat_{n},\thetaveczero\right)\right\}_{n=1}^{N}$ and $\{u\left(\Xmat_{n}\right)\psivec_{u}\left(\Xmat_{n},\thetaveczero\right)\psivec^{T}_{u}\left(\Xmat_{n},\thetaveczero\right)\}^{N}_{n=1}$, comprising  (\ref{TildeJ}) and (\ref{hatG}), respectively, are real and i.i.d. 
Furthermore, by (\ref{PsiDef}), (\ref{PsiBound}) and assumption A-\ref{A3} it can be shown that there exists a positive constant $B$ such that $\left|\left[\psivec_{u}\left(\Xmat,\thetavec\right)\right]_{k}\right|\leq{B}\sum_{l=0}^{2}\left\|\Xmat\right\|^{l}$, $k=1,\ldots,p$ and $|[\psivec_{u}\left(\Xmat;\thetavec\right)\psivec^{T}_{u}\left(\Xmat;\thetavec\right)]_{k,j}|\leq{B}\sum_{l=0}^{4}\left\|\Xmat\right\|^{l}$, ${k,j=1,\ldots,p}$. Hence, according to (\ref{Cond}), Assumption A-\ref{A4} and H\"{o}lder's inequality \cite{MeasureTheory} the expectations 
$\{{\rm{E}}[u\left(\Xmat\right)\left|\left[\psivec_{u}\left(\Xmat,\thetaveczero\right)\right]_{k}\right|;\pxtheta]\}^{p}_{k=1}$ and $\{{\rm{E}}[u^{2}\left(\Xmat\right)|[\psivec_{u}\left(\Xmat;\thetaveczero\right)\psivec^{T}_{u}\left(\Xmat;\thetaveczero\right)]_{k,j}|;\pxtheta]\}^{p}_{k,j=1}$ are finite. Therefore, by (\ref{TildeJ}), (\ref{hatG}) and Khinchine's strong law of large numbers \cite{Folland} relations (\ref{ScoreLimit}) and (\ref{CovLimit}) must hold.

Thus, by (\ref{MTRAO}), (\ref{ScoreLimit}), (\ref{CovLimit}), Assumptions A-\ref{A1}, A-\ref{A2} and the Mann-Wald theorem \cite{MannWald} we conclude that 
$\frac{1}{N}T_{u}\xrightarrow[N\rightarrow{\infty}]{\textrm{w.p. 1}}C$ under $H_{1}$, where $C$ denotes some positive constant. The relation (\ref{InfLim}) directly follows.\qed
%%%%%%%%%%%%%%%%%%%%%%%%%%%%%%%%%%%%%%%%%%%%%%%%%%%%%%%%%%%%%%%%%%%%%%%%%%%%%%%%%%%%%%%%%%%%%%%%%%%
%%%%%%%%%%%%%%%%%%%%%%%%%%%%%%%%%%%%%%%%%%%%%%%%%%%%%%%%%%%%%%%%%%%%%%%%%%%%%%%%%%%%%%%%%%%%%%%%%%%
\subsection{Proof of Proposition \ref{AChiH0}:}
\label{AChiH0Proof}
Under assumptions A-\ref{A3} and A-\ref{A4} it is shown in Lemma 5 stated in Appendix B in \cite{MTQMLFull} that the normalized score function (\ref{TildeJ}) satisfies
\begin{equation}
\label{EtaLimH0}
{\etavec}_{u}\left(\thetaveczero\right)\xrightarrow[N\rightarrow{\infty}]{D}\mathcal{N}\left(\zerovec,\Gmat_{u}\left(\thetaveczero\right)\right).
\end{equation}
Furthermore, by assumptions A-\ref{A3} and A-\ref{A4} it also follows from relation (\ref{CovLimit}) in Appendix \ref{ConsPropProof} that under $H_{0}$ ($\pxtheta=\pxthetazero$)
\begin{equation}
\label{GhatLimH0}
\hat{\Gmat}_{u}\left(\thetaveczero\right)\xrightarrow[N\rightarrow{\infty}]{\textrm{w.p. 1}}\Gmat_{u}\left(\thetaveczero\right),
\end{equation}
where according to assumption A-\ref{A5} $\Gmat_{u}\left(\thetaveczero\right)$ is non-singular. 
Hence, relation (\ref{TAsDistH0}) follows from  (\ref{EtaLimH0}), (\ref{GhatLimH0}), Slutskey's Theorem \cite{MeasureTheory}, Mann-Wald's Theorem \cite{MannWald} and the properties of quadratic forms of Gaussian random variables \cite{Kay}. \qed

%%%%%%%%%%%%%%%%%%%%%%%%%%%%%%%%%%%%%%%%%%%%%%%%%%%%%%%%%%%%%%%%%%%%%%%%%%%%%%%%%%%%%%%%%%%%%%%%%%%
%%%%%%%%%%%%%%%%%%%%%%%%%%%%%%%%%%%%%%%%%%%%%%%%%%%%%%%%%%%%%%%%%%%%%%%%%%%%%%%%%%%%%%%%%%%%%%%%%%%
\subsection{Proof of Theorem \ref{ALADTh}:}
\label{ALADThProof}
In Propositions \ref{AnorLocLem} and \ref{WConGLoc} stated below we show that $\etavec_{u}\left(\thetaveczero\right)\xrightarrow[N\rightarrow{\infty}]{D}\mathcal{N}\left(\Fmat_{u}\left(\thetaveczero\right)\hvec,\Gmat_{u}\left(\thetaveczero\right)\right)$ and $\hat{\Gmat}_{u}\left(\thetaveczero\right)\xrightarrow[N\rightarrow{\infty}]{P}\Gmat_{u}\left(\thetaveczero\right)$. Hence, by (\ref{MTRAO})
relation (\ref{AsDistH1}) follows from Slutskey's Theorem \cite{MeasureTheory}, Mann-Wald's Theorem \cite{MannWald} and the properties of quadratic forms of Gaussian random variables \cite{Kay}.\qed

The following lemmas are based on the fact that since by (\ref{LocAlt}) the parameter $\thetavec$ changes with the sample size $N$, the observations form a triangular array \cite{Serfling} (rather than a sequence) of random vectors:
\begin{equation}
\label{TriangleX}
\Xmat_{N,k},\hspace{0.2cm}k=1,\ldots,N\hspace{0.2cm},N\geq1,
\end{equation}
where $\Xmat_{N,k}$, $k=1,\ldots,N$ are i.i.d. with probability distribution $P_{\Xmatsc;\thetaveczerosc+\frac{\hvec}{\sqrt{N}}}$.
%%%
%%%
\begin{Proposition}
\label{AnorLocLem}
Assume that conditions  A-\ref{A3}-A-\ref{A7} are satisfied. Under the local alternatives (\ref{LocAlt}), the normalized score function $\etavec_{u}\left(\thetaveczero\right)$ (\ref{TildeJ}) satisfies:
\begin{equation}
\label{ASNormEtaLoc}
\etavec_{u}\left(\thetaveczero\right)\xrightarrow[N\rightarrow{\infty}]{D}\mathcal{N}\left(\Fmat_{u}\left(\thetaveczero\right)\hvec,\Gmat_{u}\left(\thetaveczero\right)\right).
\end{equation}
\end{Proposition}
\begin{proof}
By assumption A-\ref{A3}, the vector function $\psivec_{u}\left(\xvec;\thetavec\right)$ defined in (\ref{PsiDef}) is continuous in $\Thetasp$ for any $\xvec\in\XCal$. Therefore, by Identity \ref{MVTpsiID}, the normalized score function (\ref{TildeJ}) satisfies $\etavec_{u}\left(\thetaveczero\right)=\etavec_{u}\left(\thetavec\right) + \hat{\Fmat}_{u}\left(\thetavec^{*}\right)\hvec$, where $\hat{\Fmat}_{u}\left(\cdot\right)$ is defined below (\ref{EMSE}) and $\thetavec^{*}$ lies in the line segment connecting $\thetavec$ and $\thetaveczero$. Furthermore, by Lemmas \ref{AGEtathetaLoc} and \ref{WConFthetaStar} stated below in Appendix \ref{AuxLem}, $\etavec_{u}\left(\thetavec\right)\xrightarrow[N\rightarrow{\infty}]{D}\mathcal{N}\left(\zerovec,\Gmat_{u}\left(\thetaveczero\right)\right)$
and $\hat{\Fmat}_{u}\left(\thetavec^{*}\right)\xrightarrow[N\rightarrow{\infty}]{P}\Fmat_{u}\left(\thetaveczero\right)$. Therefore, the relation (\ref{ASNormEtaLoc}) follows directly from Slutskey's Theorem \cite{MeasureTheory}.
\end{proof}
%%%
%%%
\begin{Proposition}
\label{WConGLoc}
Assume that conditions A-\ref{A3}, A-\ref{A4}, A-\ref{A6}, A-\ref{A7} are satisfied. Then, under (\ref{LocAlt})
\begin{equation}
\label{ConGLoc}
\hat{\Gmat}_{u}\left(\thetaveczero\right)\xrightarrow[N\rightarrow{\infty}]{P}\Gmat_{u}\left(\thetaveczero\right).
\end{equation}
\end{Proposition}
%%%
%%%
\begin{proof}
Define a triangular array of real random variables obtained from the array (\ref{TriangleX}):
\begin{equation}
\label{YNKDef}
Y_{N,k}\triangleq{g}\left(\Xmat_{N,k}\right),\hspace{0.2cm}k=1,\ldots,N,\hspace{0.2cm}N\geq{1},
\end{equation}
where 
\begin{equation}
\label{hDef}
{g}\left(\Xmat\right)\triangleq{u}^{2}\left(\Xmat\right)\left[\psivec_{u}\left(\Xmat;\thetaveczero\right)\right]_{l}\left[\psivec_{u}\left(\Xmat;\thetaveczero\right)\right]_{m}. 
\end{equation}
Since $\Xmat_{N,k}$, $k=1,\ldots,N$ are i.i.d. and the functions $u\left(\cdot\right)$ and $\psivec_{u}\left(\cdot;\cdot\right)$ are real, then $Y_{N,k}$, $k=1,\ldots,N$ are real and i.i.d. Furthermore, define the random variable 
\begin{equation}
\label{YDef}
Y\triangleq{g}\left(\Xmat\right), 
\end{equation}
where $\Xmat$ has probability distribution $\pxthetazero$.

Let $F_{Y_{N,1}}\left(\cdot\right)$ and  $F_{Y}\left(\cdot\right)$ denote the c.d.fs of $Y_{N,1}$ and $Y$, respectively.
We show that 
\begin{equation}
\label{FYNConv}
F_{Y_{N,1}}\left(y\right)\xrightarrow[N\rightarrow{\infty}]{}F_{Y}\left(y\right)\hspace{0.2cm}\forall{y\in{C}},
\end{equation}
where $C\subseteq{\Rsp}$ denotes the set of continuity points of $F_{Y}\left(y\right)$. Let $\zeta_{Y_{N,1}}\left(t\right)$ and $\zeta_{Y}\left(t\right)$ denote the characteristic functions of $Y_{N,1}$ and $Y$, respectively. By (\ref{YNKDef}) and (\ref{YDef}) their difference satisfies
$\left|\zeta_{Y_{N,1}}\left(t\right)-\zeta_{Y}\left(t\right)\right|=\int_{\XCalsc}e^{itg\left(\xvec\right)}(f(\xvec;\thetavec)-f(\xvec;\thetaveczero))d\rho\left(\xvec\right)$, where $f(\xvec;\thetavec)\triangleq{d}\pxtheta(\xvec)/d{\rho(\xvec)}$ is the density function of $\pxtheta$ w.r.t a dominating $\sigma$-finite measure $\rho$ on $\mathcal{S}_{\XCalsc}$. Hence, by (\ref{LocAlt}), (\ref{MVTf}) and assumption A-\ref{A6}
\begin{eqnarray}
\label{ChBound}
\nonumber
\left|\zeta_{Y_{N,1}}\left(t\right)-\zeta_{Y}\left(t\right)\right|&\overset{(a)}\leq&\frac{\int_{\XCalsc}\left|\hvec^{T}\etavec(\Xmat;\thetavec^{*})\right|f\left(\xvec;\thetavec^{*}\right)d\rho\left(\xvec\right)}{\sqrt{N}}
\\\nonumber&=&\frac{{{\rm{E}}\left[\left|\hvec^{T}\etavec(\Xmat;\thetavec^{*})\right|;P_{\Xmatsc;\thetavecsc^{*}}\right]}}{\sqrt{N}}
\\\nonumber&\overset{(b)}\leq&\frac{\sqrt{{{{\rm{E}}\left[\left|\hvec^{T}\etavec(\Xmat;\thetavec^{*})\right|^{2};P_{\Xmatsc;\thetavecsc^{*}}\right]}}}}{\sqrt{N}}
\\&=&\frac{{\left\|\hvec\right\|_{\Imat_{\rm{FIM}}\left(\thetavecsc^{*}\right)}}}{\sqrt{N}},
\end{eqnarray}
where $\etavec(\xvec;\thetavec)\triangleq\nabla_{\thetavecsc}\log{f}(\xvec;\thetavec)$, $\Imat_{\rm{FIM}}\triangleq{\rm{E}}[\etavec(\xvec;\thetavec)\etavec^{T}(\xvec;\thetavec)]$ is the Fisher-information matrix \cite{KayEst}$, \thetavec^{*}$ lies in the line segment connecting $\thetavec=\thetaveczero+\frac{\hvec}{\sqrt{N}}$ and $\thetaveczero$, (a) follows from the triangle inequality, and (b) follows from H\"{o}lder's inequality \cite{MeasureTheory}. Therefore, by (\ref{ChBound}) and (\ref{hWNormB}) $\zeta_{Y_{N}}\left(t\right)\xrightarrow[N\rightarrow{\infty}]{}\zeta_{Y}\left(t\right)$ for any $t\in\Rsp$. Hence, by Theorem  11.2.2 in \cite{Lehmann}  (\ref{FYNConv}) must hold. Next, we show that 
\begin{equation}
\label{ELimY}
{\rm{E}}\left[\left|Y_{N,1}\right|;P_{Y_{N,1}}\right]\xrightarrow[N\rightarrow{\infty}]{}{\rm{E}}\left[\left|Y\right|;P_{Y}\right]<\infty.
\end{equation}
By (\ref{YNKDef}), (\ref{YDef}) and (\ref{MVTf}) it is implied that
\begin{eqnarray}
\nonumber
{\rm{E}}\left[\left|Y_{N,1}\right|;P_{Y_{N,1}}\right]&=&{\rm{E}}\left[\left|Y\right|;P_{Y}\right] 
\\\nonumber
&+&\frac{1}{\sqrt{N}}{\rm{E}}\left[\left|g\left(\Xmat\right)\right|\hvec^{T}\etavec\left(\Xmat;\thetavec^{*}\right);P_{\Xmatsc;\thetavecsc^{*}}\right],
\end{eqnarray}
where, by H\"{o}lder's inequality
\begin{equation}
\nonumber
{\rm{E}}^{2}[\left|g\left(\Xmat\right)\right|\hvec^{T}\etavec\left(\Xmat;\thetavec\right);P_{\Xmatsc;\thetavecsc}]\leq
{\rm{E}}[\left|g\left(\Xmat\right)\right|^{2};P_{\Xmatsc;\thetavecsc}]\left\|\hvec\right\|^{2}_{\Imat_{\rm{FIM}}\left(\thetavecsc\right)}.
\end{equation}
According to (\ref{PsiDef}), (\ref{hDef}), (\ref{PsiBound}), assumption A-\ref{A4}, the compactness of $\Thetasp$ and H\"{o}lder's inequality \cite{MeasureTheory}, there exists a constant $B>0$ such that $h\left(\xvec\right)\triangleq{B}u^{4}\left(\xvec\right)\sum_{r=0}^{8}\left\|\xvec\right\|^{r}\geq|g\left(\xvec\right)|^{2}$ and ${\rm{E}}\left[h\left(\Xmat\right);\pxtheta\right]$ is bounded. Hence, by (\ref{hWNormB}) we conclude that the expectation ${\rm{E}}\left[\left|g\left(\Xmat\right)\right|\hvec^{T}\etavec\left(\Xmat;\thetavec\right);P_{\Xmatsc;\thetavecsc}\right]$ must be bounded, and therefore, the relation (\ref{ELimY}) must hold.

Finally, by (\ref{GDef}), (\ref{hatG}), (\ref{YNKDef})-(\ref{FYNConv}), (\ref{ELimY}) and Lemma 15.4.1 in \cite{Lehmann} we concluded that $\hat{\Gmat}_{u}\left(\thetaveczero\right)\xrightarrow[N\rightarrow{\infty}]{P}\Gmat_{u}\left(\thetaveczero\right)$ 
\end{proof}
%%%%%%%%%%%%%%%%%%%%%%%%%%%%%%%%%%%%%%%%%%%%%%%%%%%%%%%%%%%%%%%%%%%%%%%%%%%%%%%%%%%%%%%%%%%%%%%%%%%
%%%%%%%%%%%%%%%%%%%%%%%%%%%%%%%%%%%%%%%%%%%%%%%%%%%%%%%%%%%%%%%%%%%%%%%%%%%%%%%%%%%%%%%%%%%%%%%%%%%
\subsection{Auxiliary Lemmas for Proposition \ref{AnorLocLem}:}
\label{AuxLem}
\begin{Lemma}
\label{AGEtathetaLoc}
Assume that conditions  A-\ref{A3}-A-\ref{A5} are satisfied. Under the local alternatives (\ref{LocAlt}), the normalized score function (\ref{TildeJ}) satisfies: 
\begin{equation}
\label{EtaAsNorm}
\etavec_{u}\left(\thetavec\right)\xrightarrow[N\rightarrow{\infty}]{D}\mathcal{N}\left(\zerovec,\Gmat_{u}\left(\thetaveczero\right)\right).
\end{equation}
\end{Lemma}
%%%
\begin{proof}
Define a triangular array of real random variables obtained from the array in (\ref{TriangleX}):
\begin{equation}
\label{YNKDef2}  
Y_{N,k}\triangleq{g}\left(\Xmat_{N,k}\right),\hspace{0.2cm}k=1,\ldots,N,\hspace{0.2cm}N\geq{1},
\end{equation} 
where
\begin{equation}
\label{hDef2}   
{g}\left(\xvec\right)\triangleq{u}\left(\xvec\right)\tvec^{T}\psivec_{u}\left(\xvec;\thetavec\right)
\end{equation}
and $\tvec$ is an arbitrary non-zero vector in $\Rsp^{m}$.
Since $\Xmat_{N,k}$, $k=1,\ldots,N$ are i.i.d. and the functions $u\left(\cdot\right)$ and $\psivec_{u}\left(\cdot;\cdot\right)$ are real, then $Y_{N,k}$, $k=1,\ldots,N$ are real and i.i.d.  

In the following we show some properties of the statistical moments of $Y_{N,k}$. By Identity \ref{EPsiZero}, 
\begin{equation}
{\rm{E}}\left[Y_{N,k};P_{Y_{N,1}}\right]=0.
\end{equation}
Moreover, by (\ref{GDef}), (\ref{YNKDef2}), (\ref{hDef2}), assumption A-\ref{A5} and the compactness of $\Thetasp$ there exists a constant $M>0$ such that 
\begin{equation}
\label{EY2}
{\rm{E}}\left[Y^{2}_{N,k};P_{Y_{N,1}}\right]=\left\|\tvec\right\|^{2}_{\Gmat_{u}\left(\thetavecsc\right)}\leq{M}\hspace{0.2cm}\forall\thetavec\in\Thetasp,
\end{equation}
where under (\ref{LocAlt}) we have that
\begin{equation}
\label{tGNormLim}
\left\|\tvec\right\|_{\Gmat_{u}\left(\thetavecsc\right)}\xrightarrow[N\rightarrow{\infty}]{}\left\|\tvec\right\|_{\Gmat_{u}\left(\thetaveczerosc\right)}>0.
\end{equation}
Using (\ref{YNKDef2}), (\ref{hDef2}), (\ref{PsiBound}), assumptions A-\ref{A3}, A-\ref{A4}, the compactness of $\Thetasp$ and H\"{o}lder's inequality \cite{MeasureTheory} it can be shown that the exists a constant ${M}^{\prime}>{0}$, such that 
\begin{equation}
\label{YForthB}
{\rm{E}}\left[\left|Y_{N,k}\right|^{4};P_{Y_{N,1}}\right]={\rm{E}}\left[{g}^{4}\left(\Xmat\right);\pxtheta\right]\leq{M^{\prime}}\hspace{0.2cm}\forall\thetavec\in\Thetasp.
\end{equation}

Clearly, by (\ref{EY2}) 
\begin{equation}
\label{SNEq}
s^{2}_{N}\triangleq{\sum_{k=1}^{N}{\rm{E}}[{Y}^{2}_{N,k};P_{Y_{N,1}}]}={N}\left\|\tvec\right\|^{2}_{\Gmat_{u}\left(\thetavecsc\right)}.
\end{equation}
Therefore, by (\ref{tGNormLim})-(\ref{SNEq})
\begin{equation}
\nonumber
\frac{1}{s^{4}_{N}}\sum_{k=1}^{N}{{\rm{E}}\left[\left|Y_{N,k}\right|^{4};P_{Y_{N,1}}\right]}=\frac{{{\rm{E}}\left[\left|Y_{N,1}\right|^{4};P_{Y_{N,1}}\right]}}
{N\left\|\tvec\right\|^{4}_{\Gmat_{u}\left(\thetavecsc\right)}}\xrightarrow[N\rightarrow{\infty}]{}0.
\end{equation}
Hence, by Lyapounov's central limit theorem \cite{Lehmann}, (\ref{tGNormLim}) and Slutskey's Theorem \cite{MeasureTheory} we conclude that 
\begin{equation}
\label{AsNormY1}
\frac{1}{\sqrt{N}}\sum\limits_{k=1}^{N}Y_{N,k}\xrightarrow[N\rightarrow{\infty}]{D}\mathcal{N}\left(0,\left\|\tvec\right\|^{2}_{\Gmat_{u}\left(\thetaveczerosc\right)}\right).
\end{equation}
The relation (\ref{EtaAsNorm}) follows directly from (\ref{YNKDef2}), (\ref{hDef2}), (\ref{AsNormY1}) and the Cram\'er-Wold Device \cite{Lehmann}. 
\end{proof}
%%%%%
%%%%%
\begin{Lemma}
\label{WConFthetaStar}
Assume that conditions A-\ref{A3}, A-\ref{A4}, A-\ref{A6}, and A-\ref{A7} are satisfied. Under the local alternatives (\ref{LocAlt}) 
\begin{equation}
\label{FhatConv}
\hat{\Fmat}_{u}\left(\thetavec^{*}\right)\xrightarrow[N\rightarrow{\infty}]{P}\Fmat_{u}\left(\thetaveczero\right).
\end{equation}
\end{Lemma}
\begin{proof}
Define a triangular array of real random variables obtained from the array (\ref{TriangleX}):
\begin{equation}
\label{YNKDef1}  
Y_{N,k}\triangleq{g}\left(\Xmat_{N,k};\thetavec^{*}\right),\hspace{0.2cm}k=1,\ldots,N,\hspace{0.2cm}N\geq{1},
\end{equation}
where 
\begin{equation}
\label{g1Def}  
{g}\left(\Xmat;\thetavec\right)\triangleq{u}\left(\Xmat\right)\left[\Gammamat_{u}\left(\xvec;\thetavec\right)\right]_{l,m}
\end{equation}
and $\Gammamat_{u}\left(\cdot;\cdot\right)$ is defined in (\ref{GammaDef}). Since $\Xmat_{N,k}$, $k=1,\ldots,N$ are i.i.d. and the functions $u\left(\cdot\right)$ and $\Gammamat_{u}\left(\cdot;\cdot\right)$ are real, then $Y_{N,k}$, $k=1,\ldots,N$ are real and i.i.d. Furthermore, define the random variable 
\begin{equation}
\label{YDef1}  
Y\triangleq{g}\left(\Xmat;\thetaveczero\right), 
\end{equation}
where $\Xmat$ has probability distribution $\pxthetazero$.

Let $F_{Y_{N,1}}\left(\cdot\right)$ and  $F_{Y}\left(\cdot\right)$ denote the c.d.fs of $Y_{N,1}$ and $Y$, respectively.
We show that 
\begin{equation}
\label{FYNConv1}
F_{Y_{N,1}}\left(y\right)\xrightarrow[N\rightarrow{\infty}]{}F_{Y}\left(y\right)\hspace{0.2cm}\forall{y\in{C}},
\end{equation}
where $C\subseteq{\Rsp}$ denotes the set of continuity points of $F_{Y}\left(y\right)$. Let $\zeta_{Y_{N,1}}\left(t\right)$ and $\zeta_{Y}\left(t\right)$ denote the characteristic functions of $Y_{N,1}$ and $Y$, respectively. By (\ref{YNKDef1}), (\ref{YDef1}), assumption A-\ref{A6} and Identity \ref{MVTfID} their difference satisfies:
\begin{eqnarray}
\label{Bound1}
\nonumber
\left|\zeta_{Y_{N,1}}\left(t\right)-\zeta_{Y}\left(t\right)\right|&=&\Big|{\rm{E}}\left[e^{ig\left(\Xmatsc;\thetavecsc^{*}\right)t}-e^{ig\left(\Xmatsc;\thetaveczerosc\right)t};\pxthetazero\right]
\\\nonumber&+&\frac{1}{\sqrt{N}}{\rm{E}}\left[\hvec^{T}\etavec\left(\Xmat;\thetavec^{*}\right);P_{\Xmatsc;\thetavecsc^{*}}\right]\Big|
\\\nonumber
&\overset{(a)}\leq&{\rm{E}}\left[\left|e^{ig\left(\Xmatsc;\thetavecsc^{*}\right)t}-e^{ig\left(\Xmatsc;\thetaveczerosc\right)t}\right|;\pxthetazero\right]
\\\nonumber&+&\frac{1}{\sqrt{N}}{\rm{E}}\left[\left|\hvec^{T}\etavec\left(\Xmat;\thetavec^{*}\right)\right|;P_{\Xmatsc;\thetavecsc^{*}}\right]
\\
&\overset{(b)}\leq&{\rm{E}}\left[\left|s\left(\Xmat;\thetavec^{*},\thetaveczero,t\right)\right|;\pxthetazero\right]
\\\nonumber&+&\frac{1}{\sqrt{N}}\left\|\hvec\right\|_{\Imat_{\rm{FIM}}\left(\thetavecsc^{*}\right)},
\end{eqnarray}
where $\etavec(\xvec;\thetavec)\triangleq\nabla_{\thetavecsc}\log{f}(\xvec;\thetavec)$, $\Imat_{\rm{FIM}}\triangleq{\rm{E}}[\etavec(\xvec;\thetavec)\etavec^{T}(\xvec;\thetavec)]$ is the Fisher-information matrix \cite{KayEst}, and the term $s\left(\xvec;\thetavec,\thetaveczero,t\right)\triangleq2\sin\left(\left(g\left(\xvec;\thetavec\right)-g\left(\xvec;\thetaveczero\right)\right)t\right)$.
We note that (a) follows from the triangle inequality, while (b) follows from H\"{o}lder's inequality \cite{MeasureTheory} and the definition of 
the Fisher-information matrix. Notice that $s\left(\xvec;\thetavec,\thetaveczero,t\right)$ is bounded and by (\ref{GammaDef}), (\ref{g1Def}) and A-\ref{A3} it is also continuous at $\thetavec=\thetaveczero$ for any $\xvec\in\XCal$. Therefore, since $\thetavec^{*}\xrightarrow[N\rightarrow{\infty}]{}\thetaveczero$, by the bounded convergence theorem \cite{MeasureTheory} we conclude that ${\rm{E}}\left[\left|s\left(\Xmat;\thetavec^{*},\thetaveczero,t\right)\right|;\pxthetazero\right]\xrightarrow[N\rightarrow{\infty}]{}0$ for any $t\in\Rsp$. Furthermore, by assumption A-\ref{A7} and (\ref{hWNormB}) the weighted Euclidean norm ${{\left\|\hvec\right\|^{2}_{\Imat_{\rm{FIM}}\left(\thetavecsc\right)}}}$ is bounded over $\Thetasp$, and therefore, by (\ref{Bound1}) $  \zeta_{Y_{N}}\left(t\right)\xrightarrow[N\rightarrow{\infty}]{}\zeta_{Y}\left(t\right)$ for any $t\in\Rsp$. Hence, by Theorem  11.2.2 in \cite{Lehmann} (\ref{FYNConv1}) must hold.

Next, we show that 
\begin{equation}
\label{ELimY1}
{\rm{E}}\left[\left|Y_{N,1}\right|;P_{Y_{N,1}}\right]\xrightarrow[N\rightarrow{\infty}]{}{\rm{E}}\left[\left|Y\right|;P_{Y}\right]<\infty.
\end{equation}
By (\ref{YNKDef1}), (\ref{YDef1}) and (\ref{MVTf})
\begin{eqnarray}
\label{ELimYN1}
{\rm{E}}\left[\left|Y_{N,1}\right|;P_{Y_{N,1}}\right]&=&{\rm{E}}\left[\left|g\left(\Xmat;\thetavec^{*}\right)\right|;\pxthetazero\right] 
\\\nonumber
&+&\frac{{\rm{E}}\left[\left|g\left(\Xmat;\thetavec^{*}\right)\right|\hvec^{T}\etavec\left(\Xmat;\thetavec^{*}\right);P_{\Xmatsc;\thetavecsc^{*}}\right]}{\sqrt{N}}
\end{eqnarray}
According to (\ref{GammaDef}), (\ref{g1Def}) and assumption A-\ref{A3} $g\left(\xvec,\thetavec\right)$ is continuous in $\Thetasp$ 
for any $\xvec\in\XCal$. Moreover, by (\ref{GammaDef}), (\ref{g1Def}), (\ref{GammaBound}), assumptions A-\ref{A3}, A-\ref{A4}, the compactness of $\Thetasp$ and H\"{o}lder's inequality \cite{MeasureTheory} there exist positive constants $B,C$ such that:
\begin{equation}
\label{hdef}
h\left(\xvec\right)\triangleq{B}\sum_{r=0}^{2}\left\|\xvec\right\|^{r}\geq\left|g\left(\xvec,\thetavec\right)\right|
\end{equation}
and 
\begin{equation}
\label{hIntBound}
{\rm{E}}\left[h\left(\Xmat\right);\pxtheta\right]\leq{\rm{E}}\left[h^{2}\left(\Xmat\right);\pxtheta\right]\leq{C}.
\end{equation}
Therefore, since $\thetavec^{*}\xrightarrow[N\rightarrow{\infty}]{}\thetaveczero$, by (\ref{YDef1}), (\ref{hdef}), (\ref{hIntBound}) and the dominated convergence theorem \cite{MeasureTheory} we conclude that the expectation ${\rm{E}}\left[\left|g\left(\Xmat;\thetavec^{*}\right)\right|;\pxthetazero\right]\xrightarrow[N\rightarrow{\infty}]{}{\rm{E}}\left[\left|g\left(\Xmat;\thetaveczero\right)\right|;\pxthetazero\right]\triangleq{\rm{E}}\left[\left|Y\right|;P_{Y}\right]<\infty$. We now show that the second summand in (\ref{ELimYN1}) converges to zero as $\thetavec^{*}$ approaches $\thetaveczero$. By H\"{o}lder's inequality \cite{MeasureTheory}, (\ref{hdef}) and (\ref{hIntBound})
$\frac{\left|{\rm{E}}\left[\left|g\left(\Xmatsc;\thetavecsc^{*}\right)\right|\hvec^{T}\etavecsc\left(\Xmatsc;\thetavecsc^{*}\right);P_{\xvec;\thetavecsc^{*}}\right]\right|^{2}}{\left\|\hvec\right\|^{2}_{\Imat_{\rm{FIM}}\left(\thetavecsc^{*}\right)}}\leq
{{\rm{E}}\left[\left|g\left(\Xmat;\thetavec^{*}\right)\right|^{2};P_{\Xmatsc;\thetavecsc^{*}}\right]}
\leq{{\rm{E}}\left[h^{2}\left(\Xmat\right);P_{\Xmatsc;\thetavecsc^{*}}\right]}\leq{C}$. Therefore, by (\ref{hWNormB}) we conclude that the nominator in (\ref{ELimYN1}) is bounded, and hence, the right summand in the l.h.s. of (\ref{ELimY1}) approaches zero as $N\rightarrow\infty$.

Finally, by (\ref{GammaDef}), [Fhat], (\ref{YNKDef1})-(\ref{FYNConv1}), (\ref{ELimY1}) and Lemma 15.4.1 in \cite{Lehmann} we conclude that (\ref{FhatConv}) holds.
\end{proof}

%%%%%%%%%%%%%%%%%%%%%%%%%%%%%%%%%%%%%%%%%%%%%%%%%%%%%%%%%%%%%%%%%%%%%%%%%%%%%%%%%%%%%%%%%%%%%%%%%%%
%%%%%%%%%%%%%%%%%%%%%%%%%%%%%%%%%%%%%%%%%%%%%%%%%%%%%%%%%%%%%%%%%%%%%%%%%%%%%%%%%%%%%%%%%%%%%%%%%%%

\subsection{Useful relations and identities:}
%%%
\begin{Relation}
\label{Rel1}
Assume that the parameter space $\Thetasp$ is compact. Furthermore, assume that $\muvec^{\left(u\right)}_{\Xmatsc}\left(\thetavec\right)$ and $\bSigma^{\left(u\right)}_{\Xmatsc}\left(\thetavec\right)$ are twice continuously differentiable in $\Thetasp$. Define $d\left(\xvec\right)\triangleq\sum_{r=0}^{2}\left\|\xvec\right\|^{r}$. According to (\ref{PsiDef}), (\ref{GammaDef}) the triangle inequality, the Cauchy-Schwartz inequality and its the matrix extension \cite{CS},
there exist positive constants $B_{1}$ and $B_{2}$, such that for any $(\xvec,\thetavec)\in\XCal\times\Thetasp$:
\begin{equation}
\label{PsiBound}
\left[\psivec_{u}(\xvec;\thetavec)\right]_{k}\leq{B_{1}}d\left(\xvec\right),\hspace{0.2cm}\forall{k=1,\ldots,m}
\end{equation}
and
\begin{equation}
\label{GammaBound}
\left[\Gammamat_{u}\left(\xvec;\thetavec\right)\right]_{k,l}\leq{B_{2}}d\left(\xvec\right)\hspace{0.2cm}\forall{k,l}=1,\ldots,m.
\end{equation}
\end{Relation}
%%%
%%%
\begin{Relation}
\label{RelFIM}
Assume that the parameter space $\Thetasp$ is compact. Furthermore, assume that the Fisher-Information matrix ${\Imat_{\rm{FIM}}\left(\thetavec\right)}$  \cite{KayEst} is continuous in $\Thetasp$. 
For any $\hvec\in\Rsp^{m}$ there exists a positive constant $M$ such that
\begin{equation}
\label{hWNormB}
{{\left\|\hvec\right\|^{2}_{\Imat_{\rm{FIM}}\left(\thetavecsc\right)}}}\leq{M}\hspace{0.2cm}\forall\thetavec\in\Thetasp.
\end{equation} 
\end{Relation}
%%%
%%%
\begin{Identity} 
\label{MVTpsiID}  
Assume that the vector function $\psivec_{u}\left(\xvec;\thetavec\right)$ (\ref{PsiDef}) is continuous in $\Thetasp$ for any $\xvec\in\XCal$. Furthermore, let $\thetavec=\thetavec_{0}+\frac{\hvec}{\sqrt{N}}$.
By the mean-value Theorem \cite{AdCalMult}, applied to each entry of $\psivec_{u}\left(\xvec;\thetavec\right)$ and (\ref{GammaDef}), the normalized score function $\etavec_{u}\left(\thetavec\right)$ (\ref{TildeJ}) satisfies:
\begin{equation}
\label{MVTpsi} 
\etavec_{u}\left(\thetaveczero\right)=\etavec_{u}\left(\thetavec\right) + \hat{\Fmat}_{u}\left(\thetavec^{*}\right)\hvec
\end{equation}
where $\hat{\Fmat}_{u}\left(\cdot\right)$ is defined below (\ref{EMSE}) and $\thetavec^{*}$ lies in the line segment connecting $\thetavec$ and $\thetaveczero$. 
\end{Identity}
%%%
%%%
\begin{Identity}
\label{EPsiZero}
The vector function $\psivec_{u}\left(\Xmat;\thetavec\right)$ (\ref{PsiDef}) satisfies:
\begin{equation}
\label{EPsiU}
{\rm{E}}\left[u\left(\Xmat\right)\psivec_{u}\left(\Xmat;\thetavec\right);\pxtheta\right]=\zerovec.
\end{equation}
[The proof appears in Lemma 6 in \cite{MTQMLFull}]
\end{Identity}
%%%
%%%
\begin{Identity} 
\label{MVTfID}  
Let $f\left(\xvec;\thetavec\right)\triangleq{d}\pxtheta(\xvec)/d\rho(\xvec)$ denote the density function of $\pxtheta$ w.r.t. a $\sigma$-finite measure $\rho$ on $\mathcal{S}_{\XCalsc}$.
Assume that $f\left(\xvec;\thetavec\right)$ is continuous in $\Thetasp$ $\rho$-a.e. Furthermore, let $\thetavec=\thetavec_{0}+\frac{\hvec}{\sqrt{N}}$.
By the mean-value Theorem \cite{AdCalMult} 
\begin{equation}
\label{MVTf}
f\left(\xvec;\thetavec\right)=f(\xvec;\thetaveczero)+\frac{\hvec^{T}}{\sqrt{N}}\etavec(\xvec;\thetavec^{*})f(\xvec;\thetavec^{*})\hspace{0.2cm}\rho{\textrm{-a.e.}},
\end{equation}
where $\etavec(\xvec;\thetavec)\triangleq\nabla_{\thetavecsc}\log{f}(\xvec;\thetavec)$ and $\thetavec^{*}$ lies in the line segment connecting $\thetavec$ and $\thetaveczero$. 
\end{Identity}

%%%%%%%%%%%%%%%%%%%%%%%%%%%%%%%%%%%%%%%%%%%%%%%%%%%%%%%%%%%%%%%%%%%%%%%%%%%%%%%%%%%%%%%%%%%%%%%%%%%
%%%%%%%%%%%%%%%%%%%%%%%%%%%%%%%%%%%%%%%%%%%%%%%%%%%%%%%%%%%%%%%%%%%%%%%%%%%%%%%%%%%%%%%%%%%%%%%%%%%

%%%%%%%%%%%%%%%%%%%%%%%%%%%%%%%%%%%%%%%%%%%%%%%%%%%%%%%%%%%%%%%%%%%%%%%%%%%%%%%%%%%%%%%%%%%%%%%%%%%%%%%%%%%%%%%%%%%%
%%%%%%%%%%%%%%%%%%%%%%%%%%%%%%%%%%%%%%%%%%%%%%%%%%%%%%%%%%%%%%%%%%%%%%%%%%%%%%%%%%%%%%%%%%%%%%%%%%%%%%%%%%%%%%%%%%%%

\newpage
\bibliographystyle{IEEEbib}
\bibliography{strings,refs}

\end{document}